\documentclass[letterpaper, 9 pt, conference]{ieeeconf}      % Use this line for a4
\IEEEoverridecommandlockouts                              % This command is only
\usepackage{graphics} % for pdf, bitmapped graphics files
\usepackage{graphicx}
\usepackage[export]{adjustbox}
\usepackage{epsfig} % for postscript graphics files
\usepackage{times} % assumes new font selection scheme installed
\usepackage{amsmath} % assumes amsmath package installed
\usepackage{amssymb}  % assumes amsmath package installed
\usepackage{subfigure}
\usepackage{url}
\usepackage{dsfont}
\newtheorem{theorem}{Theorem}

\newtheorem{algorithm}[theorem]{Algorithm}

\usepackage{cite}

\usepackage{color}
\usepackage{algorithm}

\newtheorem{definition}[theorem]{Definition}

{ % these settings define remarks and example as parts
}

\newcommand{\bi}{\begin{itemize}}
\newcommand{\ei}{\end{itemize}}
\newcommand{\bd}{\begin{displaymath}}
\newcommand{\ed}{\end{displaymath}}
\newcommand{\be}{\begin{eqnarray*}}
\newcommand{\ee}{\end{eqnarray*}}

\usepackage[normalem]{ulem}

\title{\LARGE \bf
Information Based Data-Driven Characterization of Stability and Influence in Power Systems}
\author{ Subhrajit Sinha, Pranav Sharma, Venkataramana Ajjarapu and  Umesh Vaidya\\
\thanks{S. Sinha is with Pacific Northwest National Laboratory, Richland, Washington. P. Sharma and V. Ajjarapu are with the Department of Electrical Engineering at Iowa State University, Ames, Iowa. U. Vaidya is with Department of Mechanical Engineering at Clemson University, Clemson, South Carolina. \tt \small {email : subhrajit.sinha@pnnl.gov}
}
}

\begin{document}
\maketitle

\begin{abstract}
Stability analysis of a power network and its characterization (voltage or angle) is an important problem in the power system community. However, these problems are mostly studied using linearized models and participation factor analysis. In this paper, we provide a purely data-driven technique for small-signal stability classification (voltage or angle stability) and influence characterization for a power network. In particular, we use Koopman operator framework for data-driven discovery of the underlying power system dynamics and then leverage the newly developed concept of information transfer for discovering the causal structure. We further use it to not only identify the influential states (subspaces) in a power network, but also to clearly characterize and classify angle and voltage instabilities. We demonstrate the efficacy of the proposed framework on two different systems, namely the 3-bus system, where we reproduce the already known results regarding the types of instabilities, and the IEEE 9-bus system where we identify the influential generators and also the generator (and its states) which contribute to the system instability, thus identifying the type of instability.
\end{abstract}

\section{Introduction}\label{section_intro}
Power systems are large networks with numerous components and complicated topology and reliable operation of the power grid is of utmost importance in today's world. As such, much research is devoted to ensure reliable and safe operation of power systems. In years gone by power systems were operated far from the operating point of bifurcation or instability, but with advancements in technology and an overall increase in load demand, the power systems are being operated closer and closer to their limit of stable operation\cite{shrestha2004congestion}. Moreover, it is well-known that the efficiency of power systems improves if the system is operated close to its limit of stable operation \cite{winter2015pushing}. This mode of operation warrants stability analysis and identification of states of the system which affect stability. Moreover, it is also of importance to identify the causal structure in power networks so that if some abnormal situation arises one can implement proper control action at the right state(s). For example, if there is a sudden change in the load, it is known that control action has to be taken at the angle variable(s) of the generator(s) \cite{kundur_stability_classification}. Hence, for the reliable operation of power networks, it is important to characterize their influence and causal structure.

Causality analysis has been a topic of research since the days of Aristotle, but there is no universally accepted definition of causality. The most commonly used notion of causality, known as Granger causality \cite{granger1969}, was first proposed for the analysis of econometric data and since has been used in various fields of research. Again, inspired by Shannon's information theory, bi-directional information \cite{Marko} and directed information \cite{IT_kramer_directedit} have been proposed as measures of causality and are used mostly in information theory. Another similar notion of causality is Schreiber's transfer entropy \cite{Schreiber}, but in \cite{sinha_IT_CDC2016} it was shown that in a dynamical system setting all these measures of causality fail to capture the correct causal structure. In \cite{sinha_IT_CDC2016,sinha_IT_ICC}, the authors provided a new definition of causality in terms of information flow between the states of a dynamical system and showed that it does capture the true causal structure of a dynamical system. The information transfer measure, defined in \cite{sinha_IT_CDC2016,sinha_IT_ICC} was used to define influence in a dynamical system and has been used in \cite{IT_influence_acc,sinha_cdc_2017_power,sinha_IT_power_transaction} to characterize influence and stability in power networks.

However, all these studies are model-based and modelling of power networks is difficult. Furthermore, stability analysis and influence characterization are usually performed in terms of participation factors \cite{participation_part1}, where one considers a linearized model of the power network. This often leads to crude approximations and may result in erroneous conclusions. 
Moreover, with the huge amount of data available nowadays, there is a strong necessity for data-driven causality and influence characterization. In this paper, we address these issues and show how data-driven causality and influence characterization can lead to a better understanding of power system operation. In particular, we use Koopman operator theoretic ideas \cite{Lasota,mezic_applied_koopmanism,sinha_equivariant_IFAC,sinha_sparse_koopman_acc,sinha_sparse_koopman_arxiv} to approximate the dynamics of the power network and use the obtained Koopman model to compute the information transfer measure. One of the main advantages of the Koopman model is the fact that it is \emph{not} a linearized model and it accounts for the nonlinearities of the underlying system. Furthermore, the Koopman model can be computed easily from time-series data and we leverage this to compute the information transfer between the various states (subspaces) of the power network, which is then used for stability classification and influence characterization in the power grid. We demonstrate the efficacy of the proposed approach on two different power networks, namely, stability characterization of the well-studied 3-bus system and the IEEE 9-bus system where we characterize both influence and stability.

The paper is organized as follows. In section \ref{section_IT} we review the concept of information transfer in a dynamical system and in section \ref{section_data_IT} we discuss the algorithm for data-driven computation of information transfer. In section \ref{section_3_bus} we characterize the small-signal stability of the IEEE 3 bus system followed by influence characterization and stability analysis of IEEE 9-bus system in section \ref{section_9_bus}.  Finally we conclude the paper in section \ref{section_conclusion}.

\section{Information Transfer}\label{section_IT}
In this section we briefly review the concept of information transfer in a dynamical system. For details see \cite{sinha_IT_CDC2016,sinha_IT_ICC}. Consider a discrete time dynamical system
{\small
\begin{eqnarray}\label{system2d}\left.
\begin{array}{ccl}
x(t+1) = F_x(x(t),y(t))+\xi_x(t)\\
y(t+1) = F_y(x(t),y(t))+\xi_y(t)
\end{array}\right\}=F(z(t),\xi(t))
\end{eqnarray}
}
where $x\in\mathbb{R}^{|x|}$, $y\in\mathbb{R}^{|y|}$ (here $|\cdot|$ denotes the dimension of $\{\cdot\}$), $z=(x^\top,y^\top)^\top$, and  $F_x : \mathbb{R}^{|x|+|y|}\to\mathbb{R}^{|x|}$ and $F_y : \mathbb{R}^{|x|+|y|}\to\mathbb{R}^{|y|}$ are assumed to be continuously differentiable and $\xi(t)$ are assumed to be i.i.d. noise. Let $\rho(z(t))$ denote the probability density of $z(t)$ at time $t$. In this paper, by information we mean the Shannon entropy of any probability density $\rho(\cdot)$. The intuition behind the definition of information transfer from state (subspace) $x$ to state (subspace) $y$ is that the total entropy of $y$ is the sum of the information transferred to $y$ from $x$ and the entropy of $y$ when $x$ is forcefully not allowed to evolve and is held constant (frozen). For this consider the modified system
\begin{eqnarray}\label{system_xfreeze}\left.
\begin{array}{ccl}
x(t+1) &=& x(t)\\ 
y(t+1) &=& F_y(x(t),y(t)) + \xi_y(t) 
\end{array}\right\}=F_{\not x}(z(t))
\end{eqnarray}
We denote by $\rho_{\not x}(y(t+1)|y(t))$ the probability density function of $y(t+1)$ conditioned on $y(t)$, with the dynamics in $x$ coordinate frozen in time going from time step $t$ to $t+1$ as in Eq. (\ref{system_xfreeze}). We have following definition of information transfer from $x\to y$ going from time step $t$ to $t+1$. 
\begin{definition}\label{IT_def}[Information transfer] \cite{sinha_IT_CDC2016,sinha_IT_ICC} The information transfer from $x$ to $y$ for the dynamical system (\ref{system2d}), as the system evolves from time $t$ to time $t+1$ (denoted by $[T_{x\to y}]_t^{t+1}$), is given by following formula
\begin{small}
\begin{eqnarray}\label{IT_formula}
[T_{x\to y}]_t^{t+1}=H(\rho(y(t+1)|y(t)))-H(\rho_{\not{x}}(y(t+1)|y(t))\label{IT}
\end{eqnarray}
\end{small}
where $H(\rho(y))=- \int_{\mathbb{R}^{|y|}} \rho(y)\log \rho(y)dy$ is the entropy of probability density function $\rho(y)$ and $H(\rho_{\not{x}}(y(t+1)|y(t))$ is the entropy of $y(t+1)$, conditioned on $y(t)$, where $x$ has been frozen as in Eq. (\ref{system_xfreeze}). 
\end{definition}

The information transfer from $x$ to $y$ depicts how the evolution of $x$ affects the evolution of $y$, that is, it gives a quantitative measurement of the influence of $x$ on $y$. With this, we have the following definition of influence in a dynamical system.

\begin{definition}[Influence]
We say that $x$ causes $y$ or $x$ influence $y$ if and only if the information transfer from $x$ to $y$ is non-zero.
\end{definition}

\subsection{Information transfer in linear dynamical systems}
Consider the following stochastic perturbed linear dynamical system 
\begin{eqnarray}
z(t+1)=Az(t)+\sigma \xi(t)\label{lti}
\end{eqnarray}
where $z(t)\in \mathbb{R}^N$, $\xi(t)$ is vector valued Gaussian random variable with zero mean and unit variance and $\sigma>0$ is a constant. We assume that the initial conditions are Gaussian distributed with covariance $\Sigma(0)$.
Since the system is linear, the distribution of the system states for all future time will remain Gaussian with covariance $\Sigma(t)$ satisfying 
\[A \Sigma(t-1)A^\top+\sigma^2 I=\Sigma(t).\]
To define the information transfer between various subspaces we introduce following notation to split the state $z$ as $z = [x^\top, y^\top]= [x_1^\top, x_2^\top , y^\top]$ and the $A$ matrix
% \footnote{For convenience of notation, we will sometime use the notation $z'$ to mean $z(t+1)$.}
% \begin{eqnarray}
% z(t+1)=\begin{pmatrix}x^{'}\\y^{'}\end{pmatrix}=\begin{pmatrix}A_x&A_{xy}\\ A_{yx}&A_{y}\end{pmatrix}\begin{pmatrix}x\\y\end{pmatrix}+\sigma \xi\label{splittingxy}.
% \end{eqnarray}

as:
\begin{eqnarray}
\begin{pmatrix}A_x&A_{xy}\\ A_{yx}&A_{y}\end{pmatrix}=\begin{pmatrix}A_{x_1}&A_{x_1x_2}& A_{x_1 y}\\A_{x_2x_1}&A_{x_2}& A_{x_2 y}\\ A_{y x_1}&A_{y x_2}& A_{y}\end{pmatrix}.\label{splittingA}
\end{eqnarray}
Based on this splitting, the covariance matrix $\Sigma$ at time instant $t$ can be written as follows: 
\begin{eqnarray}
\Sigma=\begin{pmatrix}\Sigma_x&\Sigma_{xy}\\\Sigma_{xy}^\top& \Sigma_y\end{pmatrix}=\begin{pmatrix} \Sigma_{x_1}&\Sigma_{x_1x_2}&\Sigma_{x_1 y}\\\Sigma_{x_1x_2}^\top&\Sigma_{x_2}&\Sigma_{x_2 y}\\\Sigma_{x_1y}^\top&\Sigma_{x_2y}^\top&\Sigma_{y}\end{pmatrix}.\nonumber\\
\label{sigma_dec}
\end{eqnarray}
Using the above notation, we can derive explicit expressions for information transfer in a linear dynamical system during transient and steady-state. In particular, we have the following expression for information transfer between various subspace
\begin{eqnarray}
[T_{x\to y}]_t^{t+1}=\frac{1}{2}\log \frac{|A_{yx}\Sigma^s_{y}(t)A_{yx}^\top +\sigma^2 I|}{\sigma^2}
\end{eqnarray}
where $\Sigma^s_y(t)=\Sigma_x(t)-\Sigma_{xy}(t)\Sigma_y(t)^{-1}\Sigma_{xy}(t)^\top$ is the Schur complement of $\Sigma_{y}(t)$ in the matrix $\Sigma(t)$ and

\begin{eqnarray}
[T_{x_1\to y}]_t^{t+1}=\frac{1}{2}\log \frac{|A_{yx}\Sigma^s_y(t)A_{yx}^\top +\sigma^2 I |}{|A_{yx_2}(\Sigma_y^{s})_{yx_2}(t)A_{yx_2}^\top+\sigma^2 I|}\label{transferx1y}
\end{eqnarray}

where $ (\Sigma_y^s)_{yx_2}$ is the Schur complement of $\Sigma_{y}$ in the matrix 
\[\begin{pmatrix}\Sigma_{x_2}&\Sigma_{x_2y}\\\Sigma_{x_2 y}^\top&\Sigma_y\end{pmatrix}.\]
The general expression for information transfer between scalar state $z_i$ and $z_j$ for linear network system can be derived from (\ref{transferx1y}). In particular with no loss of generality we can assume $z_i=z_1$ and $z_j=z_2$, then the expression for  $T_{z_1\to z_2}$ can be obtained from (\ref{transferx1y}) by defining 
\[x_1:=z_1,\;\;\; y=z_2, ,\;\;x_2:=(z_3,\ldots,z_N)\] 
\[ z_{\not{2}}=(z_1,z_3,z_4\ldots,z_N), z_{\not{1}\not{2}}=(z_3,z_4,\ldots,z_N).\]

For linear systems with Gaussian noise, the one step zero transfer can be characterized by looking at system matrix $A$. In particular, we have, $A_{z_jz_i}=0$, if and only if $[T_{z_i\to z_j}]_t^{t+1}= 0$ for all $t\in\mathbb{Z}_{\geq 0}$. For details see \cite{sinha_IT_CDC2016,sinha_IT_ICC}.

\subsection{Information transfer and stability of linear systems}
As discussed earlier, information transfer can be used to measure the influence of one state variable on another state. However, the state to state information transfer can also be used as an indicator of the instability of a system. In particular, we have the following theorem connecting information transfer and stability of the system matrix \cite{sinha_IT_power_transaction}.

\begin{theorem}\label{IT_stability_theorem}
Consider the linear system
\begin{eqnarray}\label{linear_sys_gaussian_noise}
x(t+1) = A x(t) + \xi(t)
\end{eqnarray}
where $x(t)\in\mathbb{R}^n$ is the state of the system, $t\in\mathbb{Z}_{\geq 0}$ is the time parameter, taking values in non-negative integers, $A\in\mathbb{R}^{n\times n}$ is the system matrix and $\xi(t)$ is a zero mean i.i.d. Gaussian noise with covariance $Q = E[\xi(t)\xi(t)^\top]$. Then the system matrix $A$ is Hurwitz if and only if all the information transfers as defined in (\ref{IT_formula}) are well defined and converge to a steady state value.
\end{theorem}

\begin{proof} 
See \cite{sinha_IT_power_transaction} for the proof.
\end{proof}

For example, consider a linear system
\begin{eqnarray}\label{sys_IT_stability}
\begin{aligned}
& x(t+1) = 0.4x(t) + 0.2y(t) + \xi_x(t)\\
& y(t+1) = \mu y(t) + \xi_y(t)
\end{aligned}
\end{eqnarray}
where $\mu\in [-0.99,0.99]$ and $\xi_x(t)$ and $\xi_y(t)$ are i.i.d. Gaussian noises of unit variance. The eigenvalues of the system are $(0.4,\mu)$ and hence as $|\mu|$ approaches one, the system approaches instability. The instability occurs due to $y$ dynamics and as $|\mu|$ increases, the entropy of $y$ increases rapidly. Hence, the steady state information transfer from $y$ to $x$ also increases rapidly as $|\mu|$ approaches one. This is shown in Fig. \ref{IT_stability}. 
\begin{figure}[htp!]
\centering
\includegraphics[scale=.3]{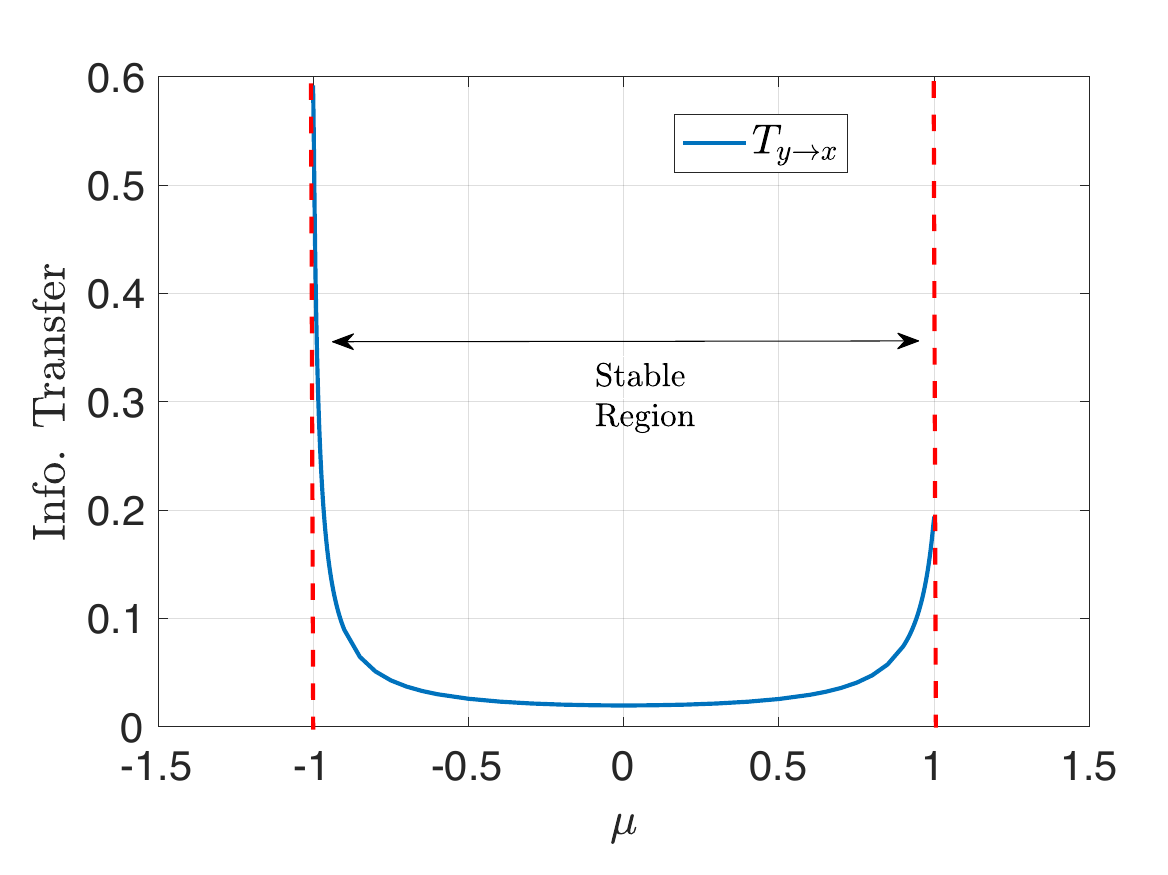}
\caption{Steady state information transfer increases rapidly as the system approaches instability.}\label{IT_stability}
\end{figure}
Conversely, if the information transfer from some state (subspace) to any other state (subspace) increases rapidly, it can be concluded that the system is approaching instability. The point to be noted is that the rapid increase in the steady-state information transfer happens when the system is still operating in the stable zone. Hence, information transfer acts as an indicator and can be used to predict the onset of instability and hence one can take preventive measures before the system becomes unstable.

\section{Data-driven computation of Information Transfer}\label{section_data_IT}

In this section, we discuss the data-driven approach to compute the information transfer for a dynamical system. For details, see \cite{sinha_IT_data_acc,sinha_IT_data_journal}.

Consider a data set obtained from a random dynamical system $z\mapsto T(z,\xi)$ 
\begin{eqnarray}
{\cal Z}  = [z_0,z_1,\ldots,z_M]
 \label{data}
\end{eqnarray}
where $z_i\in Z\subset \mathbb{R}^N$. The data-set $\{z_k\}$ can be viewed as sample path trajectory generated by random dynamical system and could be corrupted by either process or measurement noise or both. In \cite{robust_dmd_acc,sinha_robust_DMD_journal}, it was shown that Dynamic Mode Decomposition (DMD) \cite{schmid_DMD} or Extended Dynamic Mode Decomposition (EDMD) \cite{williams_EDMD,mezic_EDMD} algorithms fail to compute the Koopman operator efficiently and in \cite{robust_dmd_acc,sinha_robust_DMD_journal}, the authors provided a solution to this problem by computing a Robust Koopman operator which explicitly takes into account both process and measurement noise.

Let $\mathbf{\Psi}:Z\to \mathbb{C}^{K}$ such that
\begin{equation}
\mathbf{\Psi}(z):=\begin{bmatrix}\psi_1(z) & \psi_2(z) & \cdots & \psi_K(z)\end{bmatrix}.\label{dic_function}
\end{equation}

With this, the robust Koopman operator $(\bf K)$ can be obtained as a solution to the following optimization problem \cite{robust_dmd_acc,sinha_robust_DMD_journal}

\begin{eqnarray}\label{rob_eqv}
\min\limits_{\bf K}\parallel {\bf G}{\bf K}-{\bf A}\parallel_F+\lambda \parallel {\bf K}\parallel_F\label{regular}
\end{eqnarray}
where
\begin{eqnarray}
&&{\bf G}=\frac{1}{M}\sum_{m=0}^{M-1} \boldsymbol{\Psi}({ z}_m)^\top \boldsymbol{\Psi}({z_m })\nonumber\\
&&{\bf A}=\frac{1}{M}\sum_{m=0}^{M-1} \boldsymbol{\Psi}({ z}_m)^\top \boldsymbol{\Psi}({ z}_{m+1}),
\end{eqnarray}
${\bf K}\in\mathbb{R}^{K\times K}$ is the robust Koopman operator and $\lambda$ is the regularization parameter which depends on the bounds of the process and measurement noise.

We use ${\bf \Psi}(z) = z$, so that $\bar A=\bf K\in \mathbb{R}^{N\times N}$ is the estimated system dynamics obtained using optimization formulation (\ref{rob_eqv}). Under the assumption that the initial covariance matrix is $\bar \Sigma(0)$, the propagation of the covariance matrix under the estimated system dynamics $\bar A$ is given by 
\begin{eqnarray}
\bar \Sigma(t)=\bar A\bar \Sigma(t-1)\bar {A}^\top+\sigma^2 I\label{covariance_prop}
\end{eqnarray}
Both $\bar A$ and $\bar \Sigma$ can be decomposed according to Eqs. (\ref{splittingA}) and (\ref{sigma_dec}) and the conditional entropy $H(y_{t+1}|y_t)$ for the non-freeze case is computed using the following formula \cite{sinha_IT_CDC2016,sinha_IT_ICC}.  
\begin{eqnarray}\label{cond_entr}
H(y_{t+1}|y_t) = \frac{1}{2}\log |\bar A_{yx}\bar \Sigma_y^S(t)\bar A_{yx}^\top + \left(\frac{\lambda}{3}\right)^2 I|.
\end{eqnarray}
where $|\cdot |$ is the determinant, $\lambda$ is the bound on the process noise and $\bar \Sigma_y^S(t)$ is the Schur complement of $y$ in the covariance matrix $\bar \Sigma(t)$. In computing the entropy, we assume that the noise is i.i.d. Gaussian with covariance $\Sigma = \textnormal{diag}(\sigma^2,\cdots ,\sigma^2)$ so that one can take the bound as $\lambda = 3\sigma$, to cover the essential support of the Gaussian distribution.

For computing the dynamics when $x$ is held frozen, we modify the obtained data as follows. For simplicity, we describe the procedure for a two-state system and the method generalizes easily for the $n$-dimensional case. Let the obtained time series data be given by

\begin{eqnarray}
\mathcal{D}=\bigg[\begin{pmatrix}
x_0\\
y_0
\end{pmatrix},  \begin{pmatrix}
x_1\\
y_1
\end{pmatrix}, 
\cdots, \begin{pmatrix}
x_{M-1}\\
y_{M-1}
\end{pmatrix}\bigg] \label{data_original}
\end{eqnarray}

To replicate the effect of $x$ freeze dynamics (\ref{system_xfreeze}) we modify the original data set (\ref{data_original}) as follows.  
\begin{small}
\begin{eqnarray}
\mathcal{D}_{\not{x}}=\bigg[\left\{\begin{pmatrix}
x_0\\
y_0
\end{pmatrix}, \begin{pmatrix}
x_0\\
y_1
\end{pmatrix}\right\}, \left\{\begin{pmatrix}
x_1\\
y_1
\end{pmatrix}, \begin{pmatrix}
x_1\\
y_2
\end{pmatrix}\right\},
\cdots, \nonumber\\
\cdots ,\left\{ \begin{pmatrix}
x_{M-1}\\
y_{M-1}
\end{pmatrix}, \begin{pmatrix}
x_{M-1}\\
y_M
\end{pmatrix}\right\}\bigg]\label{mod_data}
\end{eqnarray}
\end{small}
If the original data set has $M$ data points, then the modified data set has $(2M-2)$ data points. The idea is to find the best mapping that propagates points of the form $[x_{t-1} \quad y_{t-1}]^\top$ to $[x_{t-1} \quad y_t]^\top$  (i.e., $x$ freeze) for $t = 1, 2, \hdots , M$. The estimated dynamics $\bar A_{\not x}$, when $x$ is frozen,  is calculated using the optimization formulation (\ref{regular}), but this time  applied to the data set (\ref{mod_data}). Once the frozen model is calculated, the entropy $H_{\not{x}}(y_{t+1}|y_t)$ is calculated using exactly the same procedure outline for $H(y_{t+1}|y_t)$ but this time applied to $\bar A_{\not x}$. Finally the information transfer from $x\to y$ is computed using the formula
\[T_{x\to y}=H(y_{t+1}|y_t)-H_{\not{x}}(y_{t+1}|y_t)\]

The algorithm for computing the information transfer for the identified linear system case can be  summarized as follows:
\begin{algorithm}[htp!]
\caption{Computation of Information Transfer}
\begin{enumerate}
\item{From the original data set (\ref{data_original}), compute the estimate of the system matrix $\bar A$ using the optimization formulation (\ref{rob_eqv}).}

\item{Assume $\bar \Sigma(0)$ and compute $\bar \Sigma(t)$ using Eq. (\ref{covariance_prop}). Determine $\bar A_{yx}$ and $\bar \Sigma_y^S$ to calculate the conditional entropy $H(y_{t+1}|y_t)$ using (\ref{cond_entr}).}
\item{From the original data set (\ref{data_original}) form the modified data set for the $x$ freeze dynamics as given by Eq. (\ref{mod_data}).}
\item{Follow steps (1)-(2) to compute the conditional entropy $H_{\not{x}}(y_{t+1}|y_t)$.}
\item{Compute the transfer $T_{x\to y}$ as $T_{x\to y} = H(y_{t+1}|y_t) - H_{\not{x}}(y_{t+1}|y_t)$.}
\end{enumerate}\label{algo_IT}
\end{algorithm}

\section{Stability Characterization of IEEE 3 Bus System}\label{section_3_bus}
The study of power systems' stability is a complex phenomenon. Thus, to understand the phenomenon in a simplistic way various features such as the magnitude of disturbance, time scale and causing parameters are considered in isolation. Also, the dynamic behavior of power systems is nonlinear in nature. That is why for stability analysis, usually a facile approach of linearizing the system dynamics near an equilibrium is considered and this analysis is referred to as small signal analysis of power system dynamics \cite{kundur_stability_classification}. Small signal instability can be voltage instability or angle instability, but in general, there is no unified method to classify instability as voltage or angle instability. The IEEE 3 bus system, shown in Fig. \ref{3_bus_fig}(a), is one of the very few systems where the angle instability and voltage instability has been classified \cite{ajjarapu_bifurcation}. Using the information transfer, we will show how small signal stability can be further studied to identify causation and participating states. 

\begin{figure}[htp!]
\centering
\subfigure[]{\includegraphics[scale=.45]{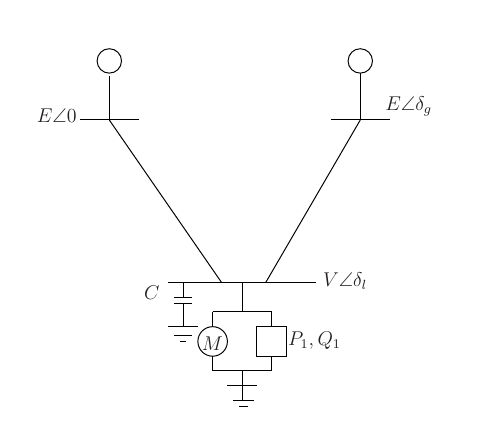}}
\subfigure[]{\includegraphics[scale=.24]{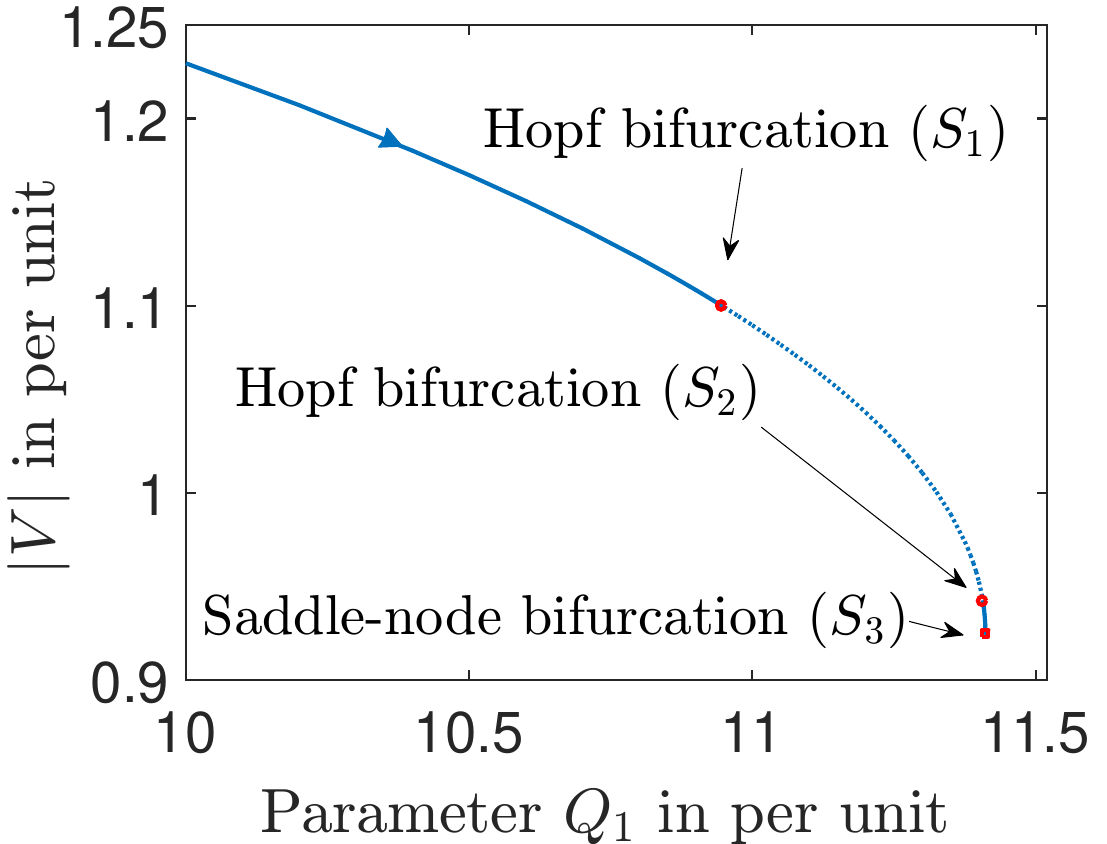}}
\caption{(a) 3-bus test case. (b) Bifurcation points on PV curve.}\label{3_bus_fig}
\end{figure}

Consider the IEEE 3-bus power network with two generators and a load, as shown in Fig. \ref{3_bus_fig}(a). The load is modelled as an induction motor in parallel with a constant $PQ$ load. The system is modelled as a four-dimensional dynamical system with the states being generator angle $(\delta_g)$, generator angular velocity $(\omega)$, the load angle $(\delta_l)$ and magnitude of load voltage $(V)$. 
The dynamic equations for the system are
\begin{eqnarray}\label{eq_del}
&& \dot{\delta_g} = \omega\\ \nonumber
&& \dot{\omega} = 16.66667\sin (\delta_l - \delta_g + 0.08727)V\\ \label{eq_omega}
&& \qquad - 0.16667 \omega + 1.88074\\ \nonumber
&& \dot{\delta_l} = 496.87181 V^2 - 166.66667\cos (\delta_l - \delta_g \\ \nonumber
&& \qquad - 0.08727)V - 666.66667\cos (\delta_l - 0.20944)V\\ \label{eq_delL}
&& \qquad 93.33333V + 33.33333Q_1 + 43.33333\\ \nonumber
&& \dot{V} = -78.76384V^2 + 26.21722\cos (\delta_l - \delta_g \\ \nonumber
&& \qquad - 0.01241)V + 104.86887\cos (\delta_l - 0.13458)V\\ \label{eq_VL}
&& \qquad + 14.52288V - 5.22876Q_1 - 7.03268
\end{eqnarray}
For a detailed analysis of the system equations, we refer the interested reader to \cite{dobson_model, ajjarapu_bifurcation}.

The above power network has three \emph{critical points}, namely $S_1$, $S_2$ and $S_3$, as shown in Fig. \ref{3_bus_fig}(b). At $S_1$ and $S_2$, a pair of imaginary eigenvalues cross the imaginary axis and at $S_3$, a real eigenvalue becomes zero. Hence, the system becomes unstable at $S_1$, remains unstable from $S_1$ to $S_2$, then regains stability after $S_2$ and again becomes unstable at $S_3$. It is known that the instability at $S_1$ is angle instability and the instability at $S_3$ is voltage instability, that is, it is the angle variable that causes instability at $S_1$ and at $S_3$ it is the voltage that is the cause for instability. 

\begin{figure}[htp!]
\centering
\includegraphics[scale=.4]{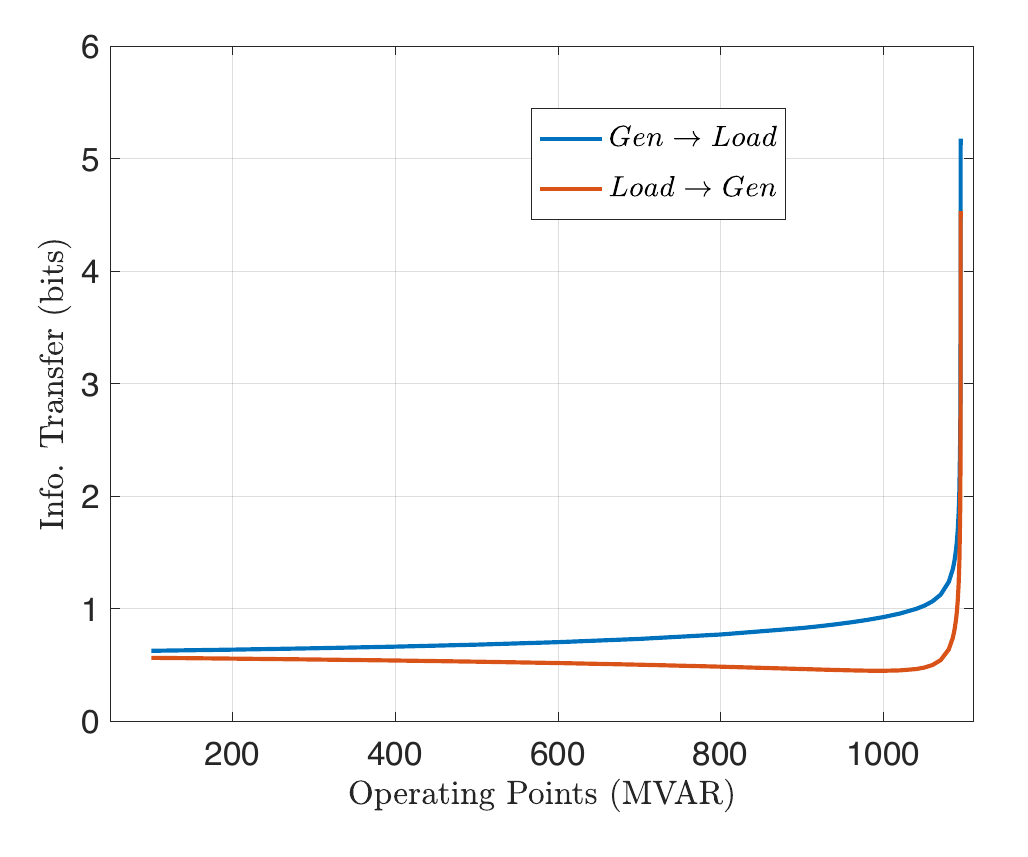}
\caption{Steady state information transfer between the generator and load subspaces over the operating points before the Hopf bifurcation $S_1$.}
\label{fig_IT_gen_load}
\end{figure}

In Fig. \ref{fig_IT_gen_load} we plot the steady-state information transfer between the generator subspace $[\delta_g,\omega]$ and the load subspace $[\delta_l,V]$, over the operating points. Time-series data of all the state variables were collected for 50-time steps at each of the 34 operating points and information transfer was computed using the algorithm \ref{algo_IT} at each of these 34 operating points. The operating points were changed by changing the reactive power $Q_1$ in (\ref{eq_del})-(\ref{eq_VL}) from 100 MVAR to 1094.6 MVAR after which the system undergoes the first Hopf bifurcation (at $S_1$) and becomes unstable. It can be seen that for most of the operating points the generator subspace has a greater influence on the load subspace than the influence of load on the generator and as the system becomes almost unstable the information transfer increases rapidly. Since the information from the generator is greater, it is reasonable to think that it is the generator subspace that is more responsible for the instability at $S_1$ than the load subspace. 

\begin{figure}[htp!]
\centering
\subfigure[]{\includegraphics[scale=.3]{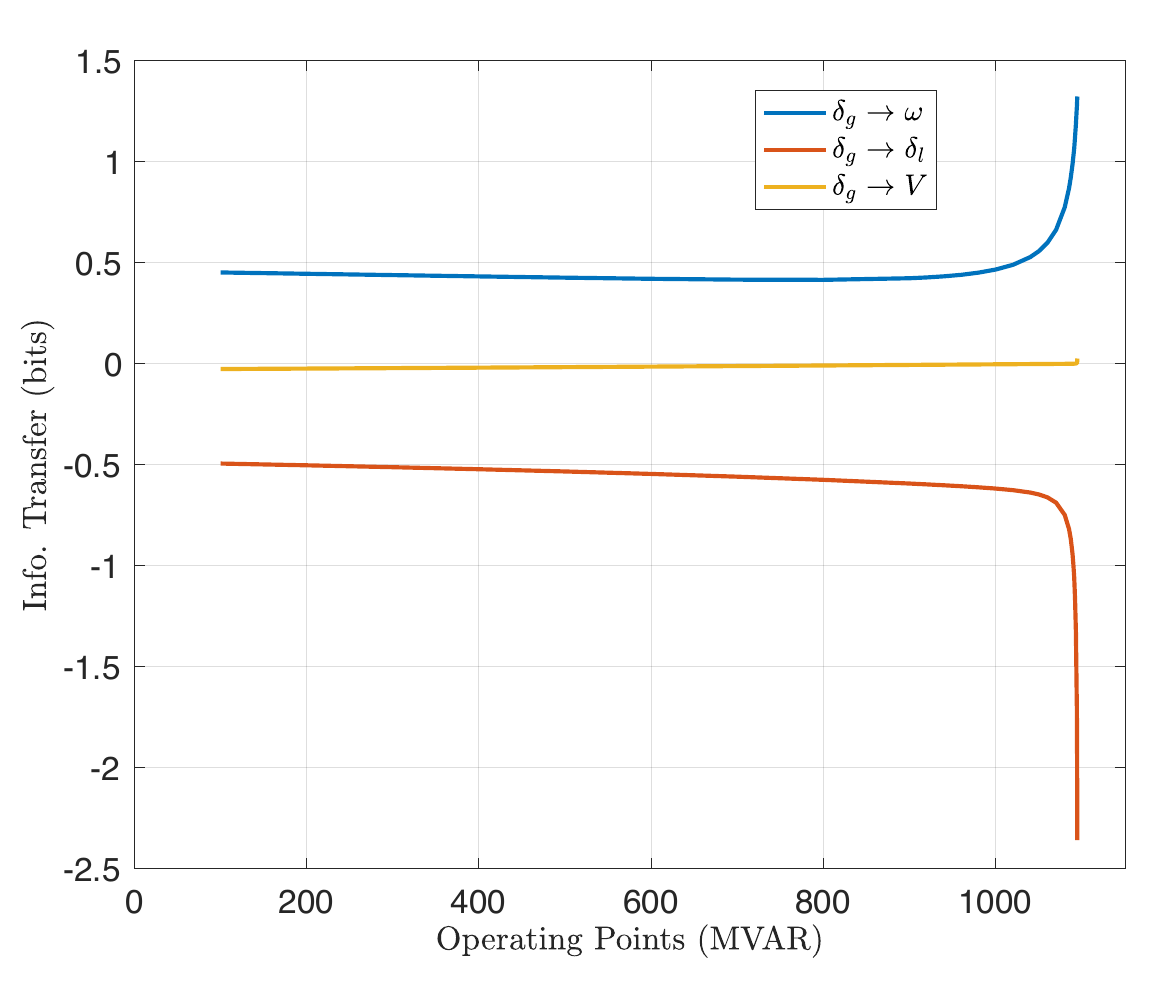}}
\subfigure[]{\includegraphics[scale=.3]{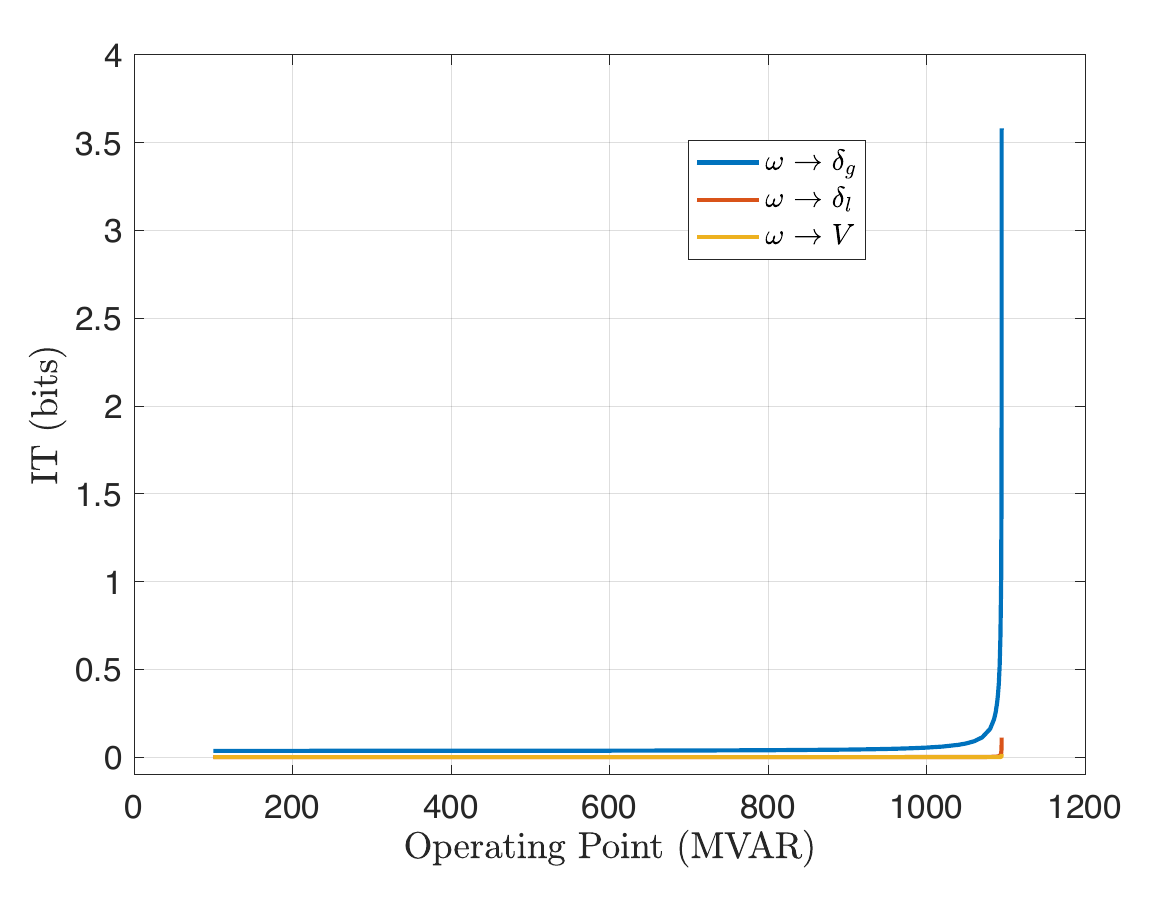}}
\subfigure[]{\includegraphics[scale=.3]{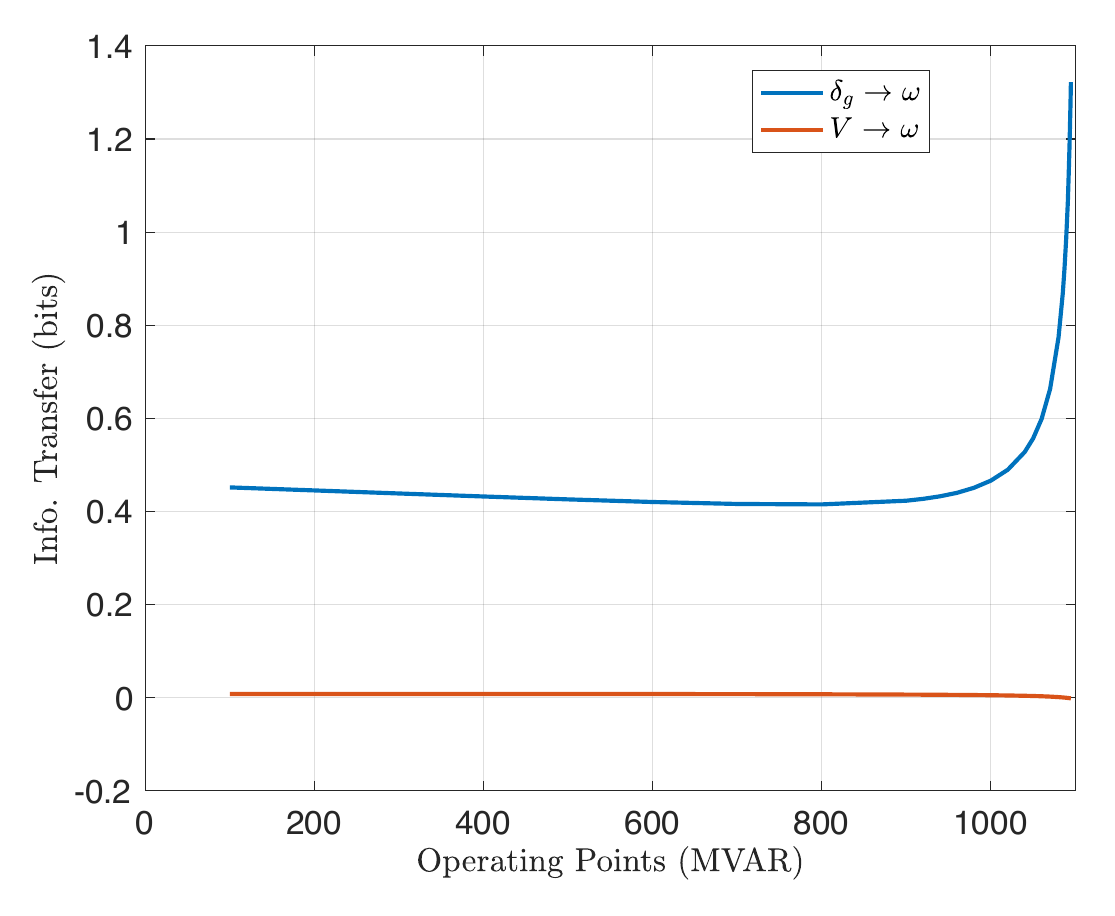}}
\caption{(a) Steady state information transfer from the angle of the generator to the other states before $S_1$. (b) Steady state information transfer from the angular speed of the generator to the other states before $S_1$. (c) Steady state information transfer to the angular speed variable from angle of the generator and voltage of the load.}\label{fig_IT_before_S1}
\end{figure}

In order to identify the state(s) responsible for instability, we now zoom into each of the subspaces. In Fig. \ref{fig_IT_before_S1}(a) and (b), we plot the information transfer from the angle of the generator $(\delta_g)$ and angular speed of the generator $(\omega)$ to all the other states. It can be observed that the absolute value of the information transfer from both the angle of the generator and angular speed of the generator increases rapidly and thus from theorem \ref{IT_stability_theorem}, we infer that the instability at $S_1$ is caused by the angle and the angular speed of the generator and thus the instability at $S_1$ is angle instability. In Fig. \ref{fig_IT_before_S1}(c) we plot the information transfer from the angle of the generator and the load voltage to the angular speed variable. It is observed that though the information transfer from the angle to the angular speed shows a sharp rise, the information transfer from the voltage to the angular speed of the generator almost remains zero. Hence, the load voltage is not making the system unstable and hence this reaffirms the statement that the instability at $S_1$ is angle instability and not voltage instability. The same can be inferred from the participation factor analysis \cite{participation_part1,verghese1982selective} and the participation of each of the states to the most unstable mode at $S_1$ is given in Table \ref{part_fact_unstable}.

{\small
\begin{table}[htp!]
\centering
\caption{Participation Factor to most unstable mode at $S_1$}\label{part_fact_unstable}
\begin{tabular}{|c|c|c|c|}
\hline
State\textbackslash  Index & Participation Factor  \\
\hline
    $\delta_g$ & $0.4825$  \\
     $\omega$ & $0.4821$ \\
     $\delta_l$ & $0.0071$ \\
     $V$ & $0.0283$\\
\hline
\end{tabular}
\end{table}}

As $Q_1$ is increased further after $S_1$, the system undergoes the first Hopf bifurcation and becomes unstable. But it regains stability at $S_2$ and as $Q_1$ is increased it undergoes saddle node bifurcation at $S_3$ and again becomes unstable. It is know that the instability at $S_3$ is voltage instability \cite{ajjarapu_bifurcation}. 
\begin{figure}[htp!]
\centering
\subfigure[]{\includegraphics[scale=.32]{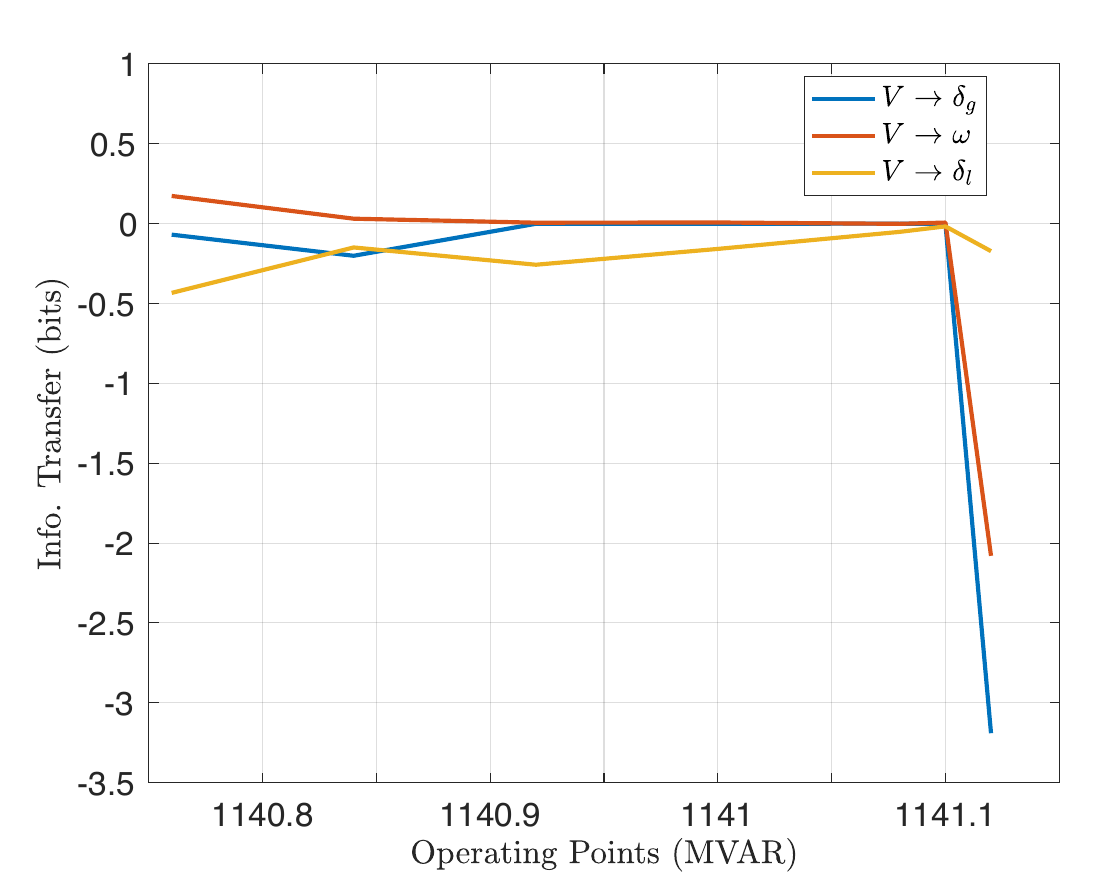}}
\subfigure[]{\includegraphics[scale=.32]{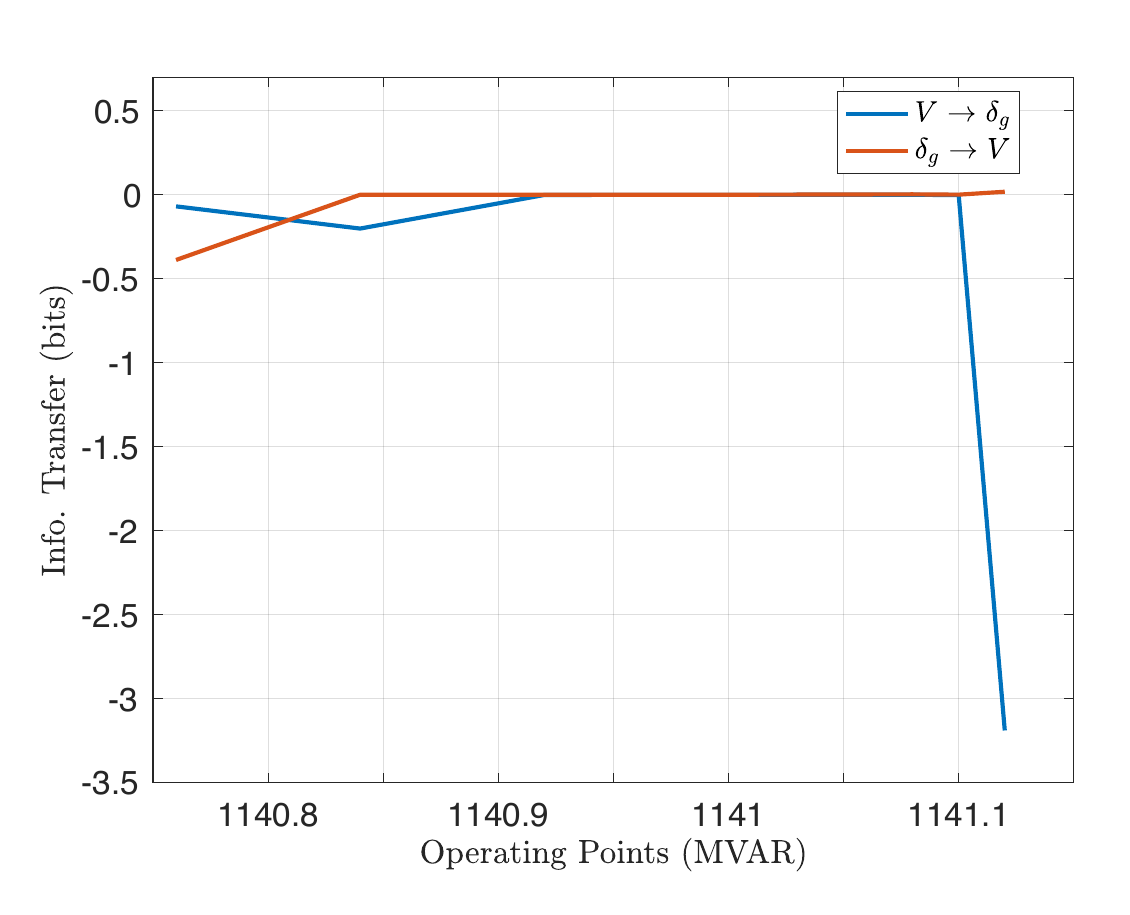}}
\caption{(a) Steady state information transfer from load voltage to all the states before $S_3$. (b) Steady state information transfer between angle of generator and load voltage before $S_3$.}\label{fig_IT_before_S3}
\end{figure}
In Fig. \ref{fig_IT_before_S3}(a) we plot the steady-state information transfer from the load voltage to all the other states through the operation points between $S_2$ and $S_3$. It can be seen that the information transfers from the load voltage increase rapidly as $Q_1$ is increased and the system approaches instability at $S_3$. Hence from theorem \ref{IT_stability_theorem} we infer that the instability at $S_3$ is caused by the load voltage. In Fig. \ref{fig_IT_before_S3}(b) we plot the information transfers between the generator angle and the load voltage. The information transfer from the angle of the generator remains almost constant throughout the operating points while the absolute value of the information transfer from load voltage increases rapidly. Thus from this, we can conclude that the instability at $S_3$ is not caused by the generator angle, but by the load voltage. Hence, information transfer identifies the type of instability in the IEEE 3 bus system.

\section{Influence Characterization of IEEE 9 Bus System}\label{section_9_bus}
In this section, we study the IEEE 9 bus system with detailed dynamic modelling for influence characterization using the concept of information transfer between states and state clusters.

\begin{figure}[htp!]
\centering
\includegraphics[scale=.375]{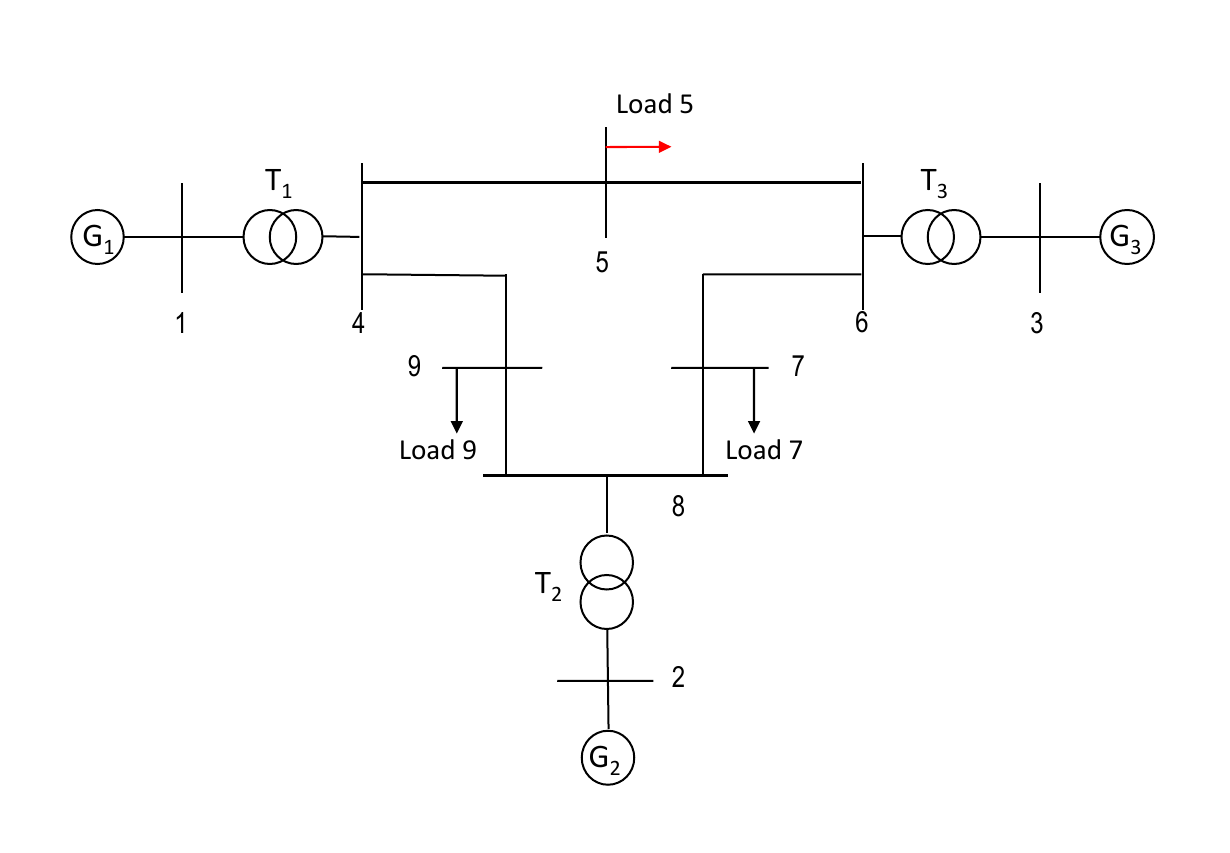}
\caption{IEEE 9 bus network}\label{9bus_diag}
\end{figure}

\subsection{IEEE 9 Bus System Modelling and Data Generation}\label{subsection_9busmodel}

To capture the dynamic behavior of a power system, dynamic algebraic equation (DAE) models are used as represented in equation (\ref{generator_dynamic_eq}). Here various components such as generator, load and controller dynamics are modeled along with network constraints. The generators are modelled as a $4^{th}$ order dynamical system with the following equations of motion.

{\small
\begin{eqnarray}\label{generator_dynamic_eq}
\begin{aligned}
&\frac{d\delta_i}{dt}  = \omega_i - \omega_s \\
&\frac{d\omega_i}{dt}  = \frac{T_{m_i}}{M_i} - \frac{E_{q_i}^{\prime} I_{q_i}}{M_i} - \frac{(X_{q_i} - X_{d_i}^{\prime})}{M_i} I_{d_i} I_{q_i} - \frac{D_i (\omega_i-\omega_s)}{M_i}  \\
&\frac{d E_{q_i}^{\prime}}{dt}  = -\frac{E_{q_i}^{\prime}}{T_{{do}_i}^{\prime}} - \frac{(X_{d_i} - X_{d_i}^{\prime})}{T_{{do}_i}^{\prime}} I_{d_i} + \frac{E_{{fd}_i}}{T_{{do}_i}^{\prime}} \\
&\frac{dE_{{fd}_i}}{dt}  = -\frac{E_{{fd}_i}}{T_{A_i}} + \frac{K_{A_i}}{T_{A_i}} (V_{{ref}_i} - V_i). 
% & \hspace{1 cm} \text{for}\;\ i = 1,\dots, n_g.
\end{aligned}
\end{eqnarray}}

Here, $\delta_i$, $\omega_i, E_{q_i}$, and $E_{{fd}_i}$ are the dynamic states of the generator and correspond to the generator rotor angle, the angular velocity of the rotor, the quadrature-axis induced emf and the emf of fast-acting exciter connected to the generator respectively. 

%The algebraic equations at the stator of the generator are: 
%%\textbf{Stator algebraic equations}
%{\small
%\begin{align}\label{stator_algebraic_eq}
%\begin{split}
%& V_i \sin(\delta_i - \theta_i) + R_{s_i} I_{d_i} - X_{q_i} I_{q_i}  = 0 \\
%& E_{q_i}^{\prime} - V_i \cos(\delta_i - \theta_i) - R_{s_i} I_{q_i} - X_{d_i}^{\prime} I_{d_i}  = 0 \\
%& \qquad \text{for}\quad  i = 1,\dots, n_g,
%\end{split}
%\end{align}}where $n_g$ is the number of generators. 

We also consider $3^{rd}$ order Power System Stabilizers (PSS) at each generator whose transfer function is given by
\begin{align}
\frac{\Delta V_{{ref}_i}(s)}{\Delta \omega_i(s)} = k_{pss} \frac{(1+sT_{num})^2}{(1+sT_{den})^2} \frac{s T_w}{1+sT_w} 
\label{eq_pss}
\end{align} 
where $k_{pss}$ is the PSS gain, $T_w$ is the time constant of wash-out filter and $T_{num}, T_{den}$ are time constants of phase-lead filter with $T_{num} > T_{den}$.

Further a $3^{rd}$ order load dynamic model is considered at load bus 5, as highlighted in Fig. \ref{9bus_diag}. For details on the modelling of the generator and the load we refer the reader to \cite{Sauer_pai_book}. The overall system dynamics is thus represented by $x \in \mathbb{R}^{24}$ order system. Further, system behavior is recorded at various load levels until it reaches voltage collapse point, as shown in Fig. \ref{9bus_pv}. At each operating point, time series data for 200-time steps is generated by perturbing system from equilibrium and this recorded data is used for influence characterization as described in section \ref{section_data_IT}.
\begin{figure}[htp!]
\centering
{\includegraphics[scale = 0.5,trim= 110 285 124 293, clip]{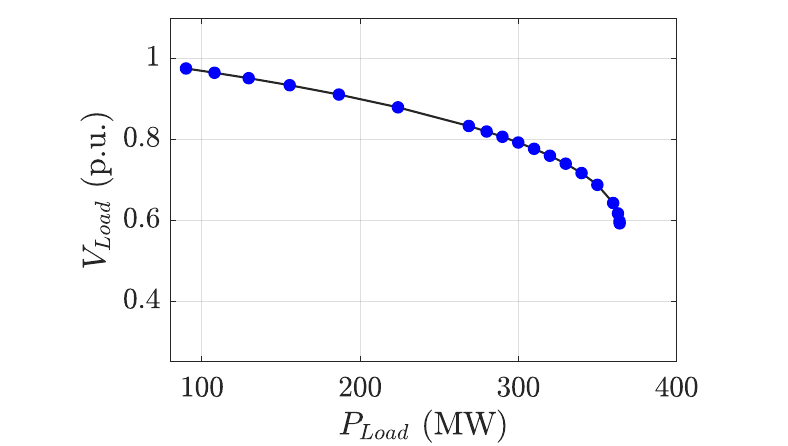}}
\caption{PV curve for IEEE 9 bus system}\label{9bus_pv}
\end{figure}

\begin{figure}[htp!]
\centering
\subfigure[]{\includegraphics[scale=.32]{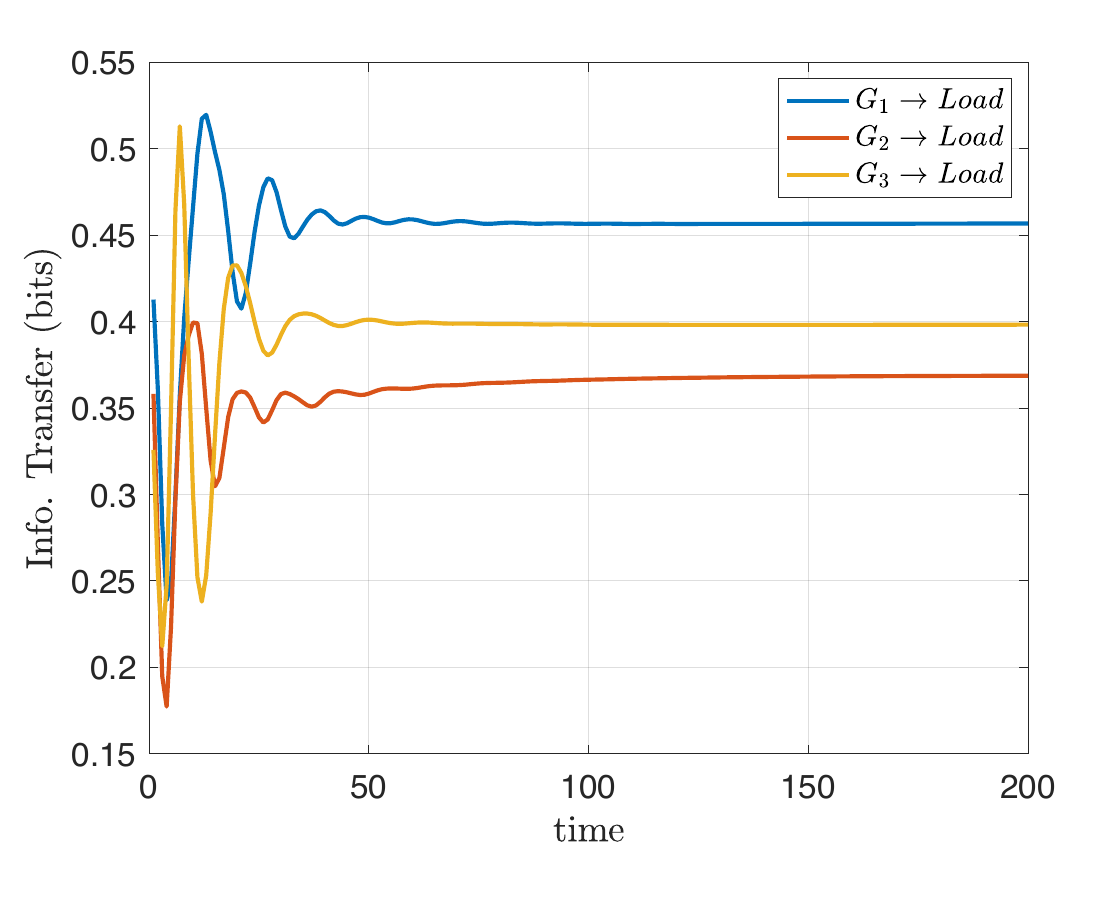}}
\subfigure[]{\includegraphics[scale=.32]{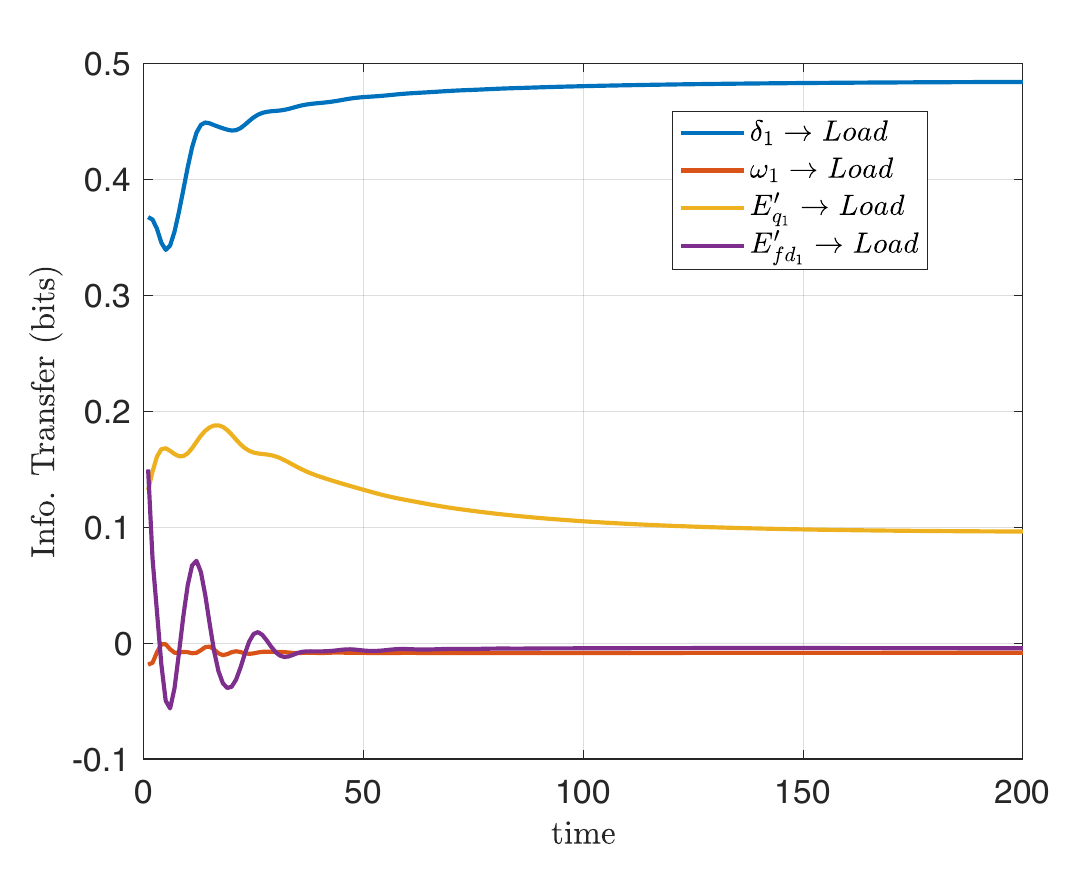}}
\caption{(a) Information transfer from each generator to load at load level 90 MW. (b) Information transfer from individual states of $G_1$ to load at load level 90 MW.}\label{fig_IT_to_load_op1}
\end{figure}

\subsection{Influence Characterization}

Satisfying the power demand for all the loads using generation sources is the primal objective of power system operation. Thus in power systems, it is crucial to identify causal interactions between the generator and load dynamics. Unlike participation factor analysis, which quantifies the participation of each state on the modes of the system, information transfer can quantify the influence of the states (and subspaces) on every other state (subspace). Since we want to identify which generator is influencing the load the most, we compute the information transfer from each generator subspace to the load subspace.

%Using the proposed notion we can compute information transfer between clusters of states in this case these clusters take the form of bundled generator and load dynamics.  At a stable operating point $P_{load} = 90 MW$, as shown in Fig. \ref{9bus_pv} information transfer from generators to load is computed.

As shown in Fig. \ref{fig_IT_to_load_op1}(a), for a load of $P_{load} = 90 MW$, $G_1$ transfer maximum information to load dynamics followed by $G_3$ and $G_2$ respectively. Once $G_1$ has been identified as the most influential generator, we now zoom into $G_1$ and identify which particular state of $G_1$ has the largest influence on the load dynamics. This zoom-in approach for information transfer computation helps greatly in the case of a power system, where thousands of state measurements increase the computation burden.

After zooming-in into individual dynamic states of $G_1$, we identify that $\delta_1$, corresponding to generator rotor angle has the highest influence on load dynamics as shown in Fig. \ref{fig_IT_to_load_op1} (b). This is in sync with the understanding of the physical behavior of a power grid in the sense that in a power network generator rotor angle is directly proportional to $P_{load}$ and hence has the highest influence on load dynamics. 

\begin{figure}[htp!]
\centering
\subfigure[]{\includegraphics[scale=.32]{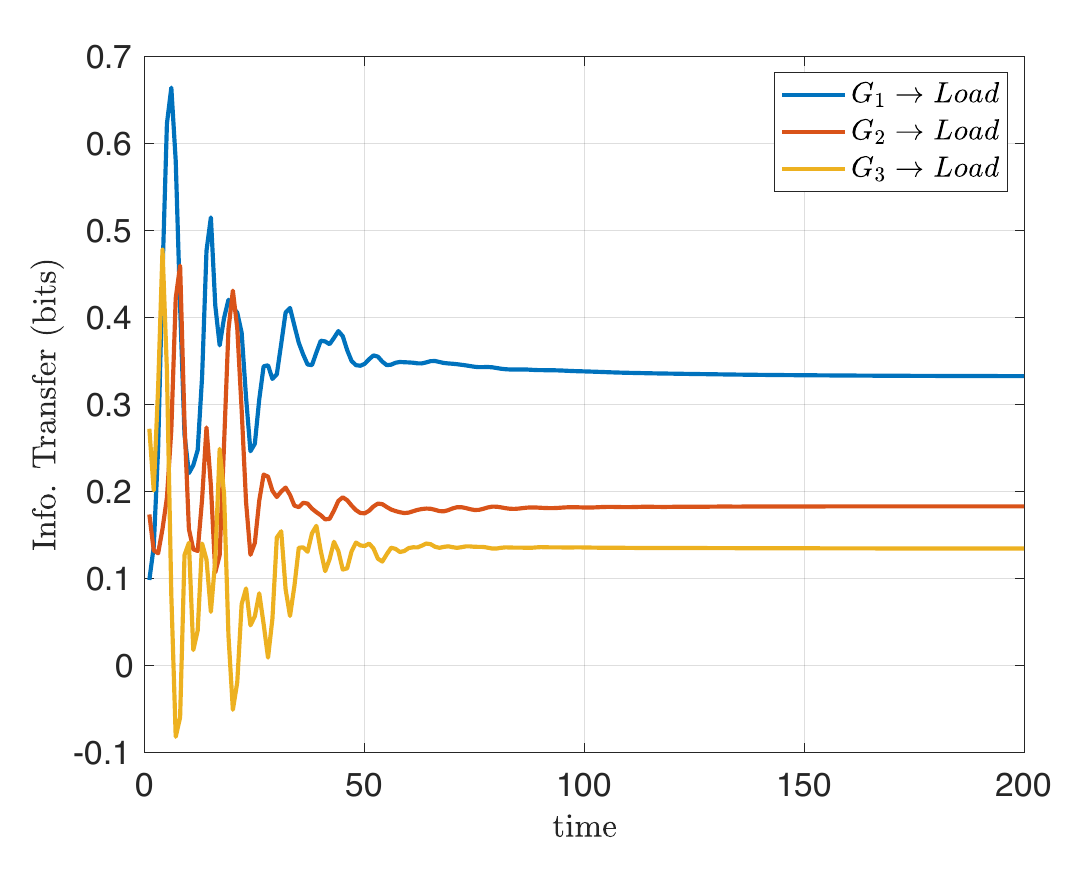}}
\subfigure[]{\includegraphics[scale=.32]{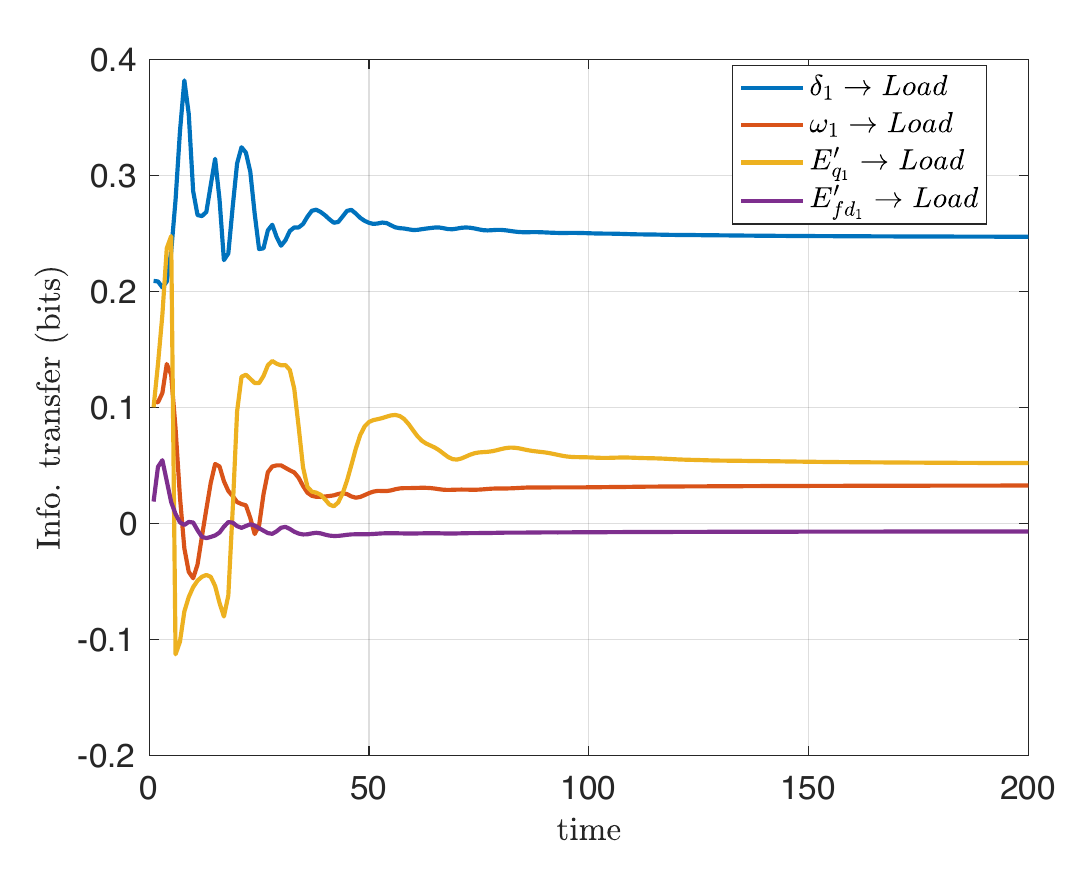}}
\caption{(a) Information transfer from each generator to load at load level 364.1 MW. (b) Information transfer from individual states of $G_1$ to load at load level 364.1 MW.}\label{fig_IT_to_load_op19}
\end{figure}
Similarly, the influence is quantified at an operating point  close to voltage collapse point for $P_{load} = 364.1 MW$. As shown in Fig. \ref{fig_IT_to_load_op19} (a), $G_1$ has maximum influence on load dynamics followed by $G_2$ and $G_3$ respectively. In Fig. \ref{fig_IT_to_load_op19} (b), we zoom in into $G_1$ to identify the most influential generator state for load dynamics, where again the generator angle state $\delta_1$ has the highest influence on the load dynamics and these results are concurrent with the expected behavior of a power system.

\subsection{Stability Characterization}

\begin{figure}[htp!]
\centering
\subfigure[]{\includegraphics[scale = 0.5,trim= 105 268 124 280, clip]{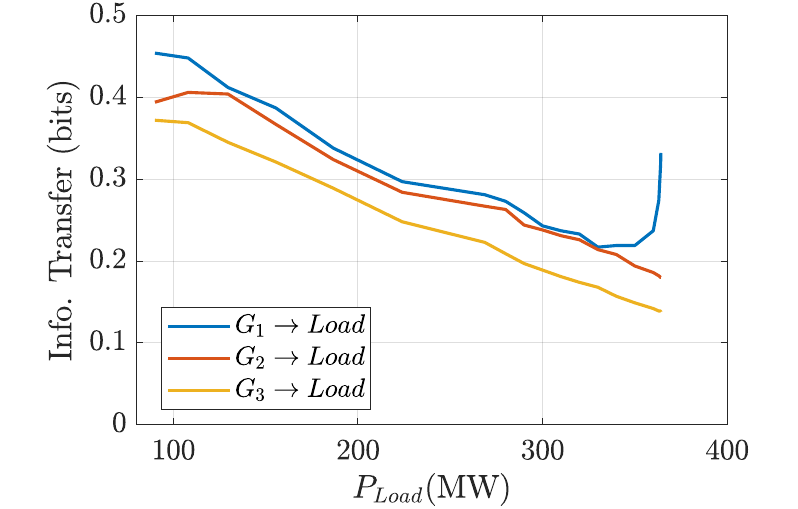}}
\subfigure[]{\includegraphics[scale = 0.5,trim= 105 258 124 270, clip]{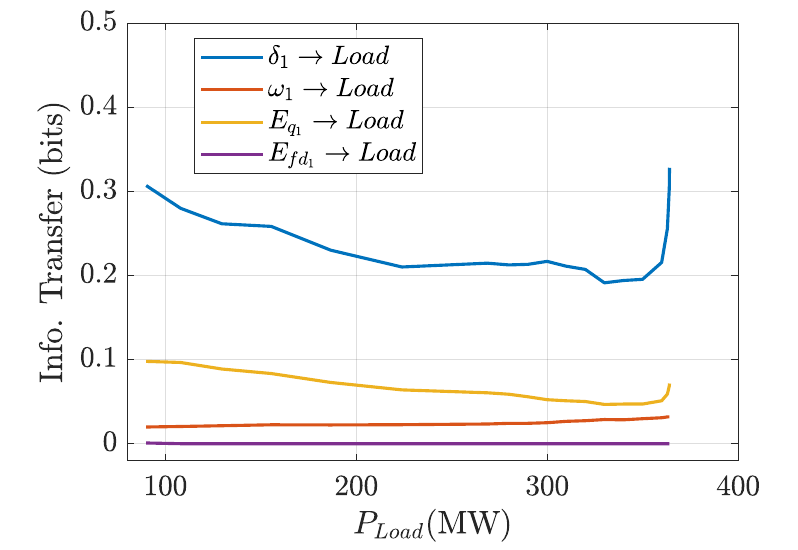}}
\caption{Steady state information transfer over operating points (a) from each generator to load (b) from states of $G_1$ to the load. }\label{IT_9bus_over_pv}
\end{figure}
The measure of information transfer can also be used to identify any approaching instabilities in the system. In order to identify the approaching instabilities, we compare the steady-state information transfer from the generators to the load over the operating points.  As shown in Fig. \ref{IT_9bus_over_pv}(a), the information transfer is computed from each generator to the load as the system loading ($P_{load}$) increases for the operating points shown in Fig. \ref{9bus_pv}. We observe that as $P_{load}$ is increased the information transfer from $G_1$ starts to increase rapidly. This indicates that the system is approaching instability (from theorem \ref{IT_stability_theorem}). Furthermore, only $G_1$ shows a sharp increase in information transfer close to the instability ($P_{load} = 364.1 MW$). Hence, it can be inferred that $G_1$ has the highest influence on approaching system stress close to the collapse points. 

Furthermore, we can also compute information transfer from the individual states of the component of interest (in this case $G_1$) to system load. As shown in Fig. \ref{IT_9bus_over_pv}(b), the state corresponding to generator angle ($\delta_1$) has the highest influence near the system collapse point. Hence, it can be inferred that the system approaches a rotor angle instability and it is the angle variable of $G_1$ that is most responsible for the instability.

It is important to note that, a data-driven information transfer thus computed, eliminates the need for eigenvalue and participation factor-based analysis for a given system. Furthermore, it provides a direct correlation between stability and causality (influence) in power systems. Also, this zoom-in approach reduces the computational burden for large power networks with thousands of states.

\section{Conclusions}\label{section_conclusion}
In this paper, we address the problem of influence and stability characterization in a power network using only time-series data of the dynamics states. In particular, we leverage the Koopman operator framework to learn the underlying dynamics and use it to compute the information transfer measure to characterize the influence of any state (subspace) on any other state (subspace). Furthermore, we use the computed information transfer metric to classify the kind of instability that a power network experiences as the network is stressed and approaches collapse. We demonstrate our proposed method first on the well-studied 3-bus system and this acts as a proof of concept. Furthermore, we analyze the causal structure of the IEEE 9-bus system, where we identify the influential states and the generators which influence the load and also classify the stability of the 9-bus system as it approaches system collapse.

\bibliographystyle{IEEEtran}
%\bibliography{ref,ref1,subhrajit_ref,subhrajit_power2}
\bibliography{ref_IT_power}

\begin{thebibliography}{10}
\providecommand{\url}[1]{#1}
\csname url@samestyle\endcsname
\providecommand{\newblock}{\relax}
\providecommand{\bibinfo}[2]{#2}
\providecommand{\BIBentrySTDinterwordspacing}{\spaceskip=0pt\relax}
\providecommand{\BIBentryALTinterwordstretchfactor}{4}
\providecommand{\BIBentryALTinterwordspacing}{\spaceskip=\fontdimen2\font plus
\BIBentryALTinterwordstretchfactor\fontdimen3\font minus
  \fontdimen4\font\relax}
\providecommand{\BIBforeignlanguage}[2]{{%
\expandafter\ifx\csname l@#1\endcsname\relax
\typeout{** WARNING: IEEEtran.bst: No hyphenation pattern has been}%
\typeout{** loaded for the language `#1'. Using the pattern for}%
\typeout{** the default language instead.}%
\else
\language=\csname l@#1\endcsname
\fi
#2}}
\providecommand{\BIBdecl}{\relax}
\BIBdecl

\bibitem{shrestha2004congestion}
G.~Shrestha and P.~Fonseka, ``Congestion-driven transmission expansion in
  competitive power markets,'' \emph{IEEE Transactions on Power Systems},
  vol.~19, no.~3, pp. 1658--1665, 2004.

\bibitem{winter2015pushing}
W.~Winter, K.~Elkington, G.~Bareux, and J.~Kostevc, ``Pushing the limits:
  Europe's new grid: Innovative tools to combat transmission bottlenecks and
  reduced inertia,'' \emph{IEEE Power and Energy Magazine}, vol.~13, no.~1, pp.
  60--74, 2015.

\bibitem{kundur_stability_classification}
P.~Kundur, J.~Paserba, V.~Ajjarapu, G.~Andersson, A.~Bose, C.~Canizares,
  N.~Hatziargyriou, D.~Hill, A.~Stankovic, C.~Taylor \emph{et~al.},
  ``Definition and classification of power system stability ieee/cigre joint
  task force on stability terms and definitions,'' \emph{IEEE transactions on
  Power Systems}, vol.~19, no.~3, pp. 1387--1401, 2004.

\bibitem{granger1969}
C.~W. Granger, ``Investigating causal relations by econometric models and
  cross-spectral methods,'' \emph{Econometrica: Journal of the Econometric
  Society}, pp. 424--438, 1969.

\bibitem{Marko}
H.~Marko, ``The bidirectional communication theory- a generalization. of
  information theory,'' \emph{{IEEE Transaction on Communications}}, vol. Com -
  21, no. 12, pp. 1345--1351, Dec, 1973.

\bibitem{IT_kramer_directedit}
G.~Kramer, ``Directed information for channels with feedback,'' in \emph{{PhD
  Thesis}}, Swiss Federal Institute of Technology Zurich, 1998.

\bibitem{Schreiber}
T.~Schreiber, ``Measuring information transfer,'' \emph{{Physical Review
  Letters}}, vol. 85, no. 2, pp. 461--464, July, 2000.

\bibitem{sinha_IT_CDC2016}
S.~Sinha and U.~Vaidya, ``Causality preserving information transfer measure for
  control dynamical system,'' \emph{IEEE Conference on Decision and Control},
  pp. 7329--7334, 2016.

\bibitem{sinha_IT_ICC}
------, ``On information transfer in discrete dynamical systems,'' \emph{Indian
  Control Conference}, pp. 303--308, 2017.

\bibitem{IT_influence_acc}
U.~Vaidya and S.~Sinha, ``Information-based measure for influence
  characterization in dynamical systems with applications,'' pp. 7147--7152,
  2016.

\bibitem{sinha_cdc_2017_power}
S.~Sinha, P.~Sharma, U.~Vaidya, and V.~Ajjarapu, ``Identifying causal
  interaction in power system: Information-based approach,'' in \emph{Decision
  and Control (CDC), 2017 IEEE 56th Annual Conference on}.\hskip 1em plus 0.5em
  minus 0.4em\relax IEEE, 2017, pp. 2041--2046.

\bibitem{sinha_IT_power_transaction}
------, ``On information transfer-based characterization of power system
  stability,'' \emph{IEEE Transactions on Power Systems}, vol.~34, no.~5, pp.
  3804--3812, 2019.

\bibitem{participation_part1}
I.~J. Perez-Arriaga, G.~C. Verghese, and F.~C. Schweppe, ``Selective modal
  analysis with applications to electric power systems, part 1: Heuristic
  introduction,'' \emph{IEEE Transactions on Power Apparatus and Systems}, vol.
  101, no.~9, pp. 3117--3125, 1982.

\bibitem{Lasota}
A.~Lasota and M.~C. Mackey, \emph{Chaos, Fractals, and Noise: Stochastic
  Aspects of Dynamics}.\hskip 1em plus 0.5em minus 0.4em\relax New York:
  Springer-Verlag, 1994.

\bibitem{mezic_applied_koopmanism}
M.~Budi{\v{s}}i{\'c}, R.~Mohr, and I.~Mezi{\'c}, ``Applied koopmanism,''
  \emph{Chaos: An Interdisciplinary Journal of Nonlinear Science}, vol.~22,
  no.~4, p. 047510, 2012.

\bibitem{sinha_equivariant_IFAC}
S.~Sinha, S.~P. Nandanoori, and E.~Yeung, ``Koopman operator methods for global
  phase space exploration of equivariant dynamical systems,''
  \emph{IFAC-PapersOnLine}, vol.~53, no.~2, pp. 1150--1155, 2020.

\bibitem{sinha_sparse_koopman_acc}
S.~Sinha, U.~Vaidya, and E.~Yeung, ``On computation of koopman operator from
  sparse data,'' in \emph{2019 American Control Conference (ACC)}.\hskip 1em
  plus 0.5em minus 0.4em\relax IEEE, 2019, pp. 5519--5524.

\bibitem{sinha_sparse_koopman_arxiv}
------, ``On few shot learning of dynamical systems: A koopman operator
  theoretic approach,'' \emph{arXiv preprint arXiv:2103.04221}, 2021.

\bibitem{sinha_IT_data_acc}
S.~Sinha and U.~Vaidya, ``Data-driven approach for inferencing causality and
  network topology,'' in \emph{2018 Annual American Control Conference
  (ACC)}.\hskip 1em plus 0.5em minus 0.4em\relax IEEE, 2018, pp. 436--441.

\bibitem{sinha_IT_data_journal}
------, ``On data-driven computation of information transfer for causal
  inference in discrete-time dynamical systems,'' \emph{Journal of Nonlinear
  Science}, vol.~30, no.~4, pp. 1651--1676, 2020.

\bibitem{robust_dmd_acc}
S.~Sinha, B.~Huang, and U.~Vaidya, ``Robust approximation of koopman operator
  and prediction in random dynamical systems,'' in \emph{2018 Annual American
  Control Conference (ACC)}.\hskip 1em plus 0.5em minus 0.4em\relax IEEE, 2018,
  pp. 5491--5496.

\bibitem{sinha_robust_DMD_journal}
------, ``On robust computation of koopman operator and prediction in random
  dynamical systems,'' \emph{Journal of Nonlinear Science}, vol.~30, no.~5, pp.
  2057--2090, 2020.

\bibitem{schmid_DMD}
P.~J. Schmid, ``Dynamic mode decomposition of numerical and experimental
  data,'' \emph{Journal of fluid mechanics}, vol. 656, pp. 5--28, 2010.

\bibitem{williams_EDMD}
M.~O. Williams, I.~G. Kevrekidis, and C.~W. Rowley, ``A data--driven
  approximation of the koopman operator: Extending dynamic mode
  decomposition,'' \emph{Journal of Nonlinear Science}, vol.~25, no.~6, pp.
  1307--1346, 2015.

\bibitem{mezic_EDMD}
M.~Korda and I.~Mezi{\'c}, ``On convergence of extended dynamic mode
  decomposition to the koopman operator,'' \emph{Journal of Nonlinear Science},
  vol.~28, no.~2, pp. 687--710, 2018.

\bibitem{ajjarapu_bifurcation}
V.~Ajjarapu and B.~Lee, ``Bifurcation theory and its application to nonlinear
  dynamical phenomena in an electrical power system,'' \emph{IEEE Transactions
  on Power Systems}, vol.~7, no.~1, pp. 424--431, 1992.

\bibitem{dobson_model}
I.~Dobson, H.-D. Chiang, J.~S. Thorp, and L.~Fekih-Ahmed, ``A model of voltage
  collapse in electric power systems,'' in \emph{Decision and Control, 1988.,
  Proceedings of the 27th IEEE Conference on}.\hskip 1em plus 0.5em minus
  0.4em\relax IEEE, 1988, pp. 2104--2109.

\bibitem{verghese1982selective}
G.~Verghese, I.~Perez-Arriaga, and F.~Schweppe, ``Selective modal analysis with
  applications to electric power systems, part ii: The dynamic stability
  problem,'' \emph{IEEE Transactions on Power Apparatus and Systems}, no.~9,
  pp. 3126--3134, 1982.

\bibitem{Sauer_pai_book}
P.~W. Sauer and M.~Pai, ``Power system dynamics and stability,'' \emph{Urbana},
  vol.~51, p. 61801, 1997.

\end{thebibliography}

\end{document}